\documentclass[12pt, a4paper]{amsart}
\usepackage{amsmath,amssymb,amscd,amsthm,amsfonts,amstext,amsbsy,mathrsfs,
hyperref,upgreek,mathtools,stmaryrd,enumitem,bbm,ifsym} 

\usepackage[textsize=footnotesize,color=blue!40, bordercolor=white]{todonotes}
\usepackage[a4paper, margin=1.2in]{geometry}

\usepackage{MnSymbol} 
\usepackage{verbatim}

\def\tinyskip{\vspace{0.0cm plus 0.1cm}}

\newtheorem{theorem}{Theorem}[section]

\newtheorem{claim}[theorem]{Claim}
\newtheorem*{claim*}{Claim}

\newtheorem*{subclaim*}{Subclaim}
 
\newtheorem*{step*}{Step}

\theoremstyle{definition}
\newtheorem{definition}[theorem]{Definition}
\newtheorem{to-do}[theorem]{Todo}

\theoremstyle{remark}
\newtheorem{remark}[theorem]{Remark}

\newenvironment{enumerate-(a)}{\begin{enumerate}[label={\upshape
(\alph*)}, leftmargin=2pc]}{\end{enumerate}}

\newenvironment{enumerate-(a)-r}{\begin{enumerate}[label={\upshape
(\alph*)}, leftmargin=2pc,resume]}{\end{enumerate}}

\newenvironment{enumerate-(A)}{\begin{enumerate}[label={\upshape
(\Alph*)}, leftmargin=2pc]}{\end{enumerate}}

\newenvironment{enumerate-(A)-r}{\begin{enumerate}[label={\upshape
(\Alph*)}, leftmargin=2pc,resume]}{\end{enumerate}}

\newenvironment{enumerate-(i)}{\begin{enumerate}[label={\upshape
(\roman*)}, leftmargin=2pc]}{\end{enumerate}}

\newenvironment{enumerate-(i)-r}{\begin{enumerate}[label={\upshape
(\roman*)}, leftmargin=2pc,resume]}{\end{enumerate}}

\newenvironment{enumerate-(I)}{\begin{enumerate}[label={\upshape
(\Roman*)}, leftmargin=2pc]}{\end{enumerate}}

\newenvironment{enumerate-(I)-r}{\begin{enumerate}[label={\upshape
(\Roman*)}, leftmargin=2pc,resume]}{\end{enumerate}}

\newenvironment{enumerate-(1)}{\begin{enumerate}[label={\upshape
(\arabic*)}, leftmargin=2pc]}{\end{enumerate}}

\newenvironment{enumerate-(1)-r}{\begin{enumerate}[label={\upshape
(\arabic*)}, leftmargin=2pc,resume]}{\end{enumerate}}

\begin{document}

\author{Sanjay Jain, Bakhadyr Khoussainov, \\
        Philipp Schlicht and Frank Stephan}
\address{Sanjay Jain, School of Computing, National University of Singapore, 
Singapore 117417, Republic of Singapore}
\email{sanjay@comp.nus.edu.sg}

\address{Bakhadyr Khoussainov, Computer Science Department, The
University of Auckland, New Zealand}
\email{bmk@cs.auckland.ac.nz}

\address{Philipp Schlicht, School of Mathematics, University of Bristol, 
University Walk, 
Bristol, 
BS8 1TW, United Kingdom} 
\email{philipp.schlicht@bristol.ac.uk} 

\address{Frank Stephan, Department of Mathematics, National University
of Singapore, Block S17, 10 Lower Kent Ridge Road, Singapore 119076,
Republic of Singapore}
\email{fstephan@comp.nus.edu.sg}

\date{\today}

\title{The isomorphism problem for tree-automatic ordinals with addition}

\begin{abstract}{This paper studies tree-automatic ordinals (or
equivalently, well-founded linearly ordered sets) together with the
ordinal addition operation $+$. Informally, these are ordinals such
that their elements are coded by finite trees for which the linear
order relation of the ordinal and the ordinal addition operation can
be determined by tree automata. 
We describe an algorithm that, given two tree-automatic ordinals with
the ordinal addition operation, decides if the ordinals are isomorphic.}
\end{abstract} 
\maketitle

\section{Motivation of the problem}

\noindent
Delhomm\'e proved that any ordinal presented by finite automata is
strictly less than $\omega^{\omega}$ \cite{D1}. Ordinals presented by
automata are called word-automatic ordinals; precise definitions will
be provided in the next section and can be found in the relevant
literature \cite{BG11, BG04, D1, FT12, FT13, R08, R19}.
Later Khoussainov, Rubin and Stephan proved
that the Cantor-Bendixson rank of any word-automatic ordinal is finite
\cite{KRS05}. In particular, this implies that any word-automatic
ordinal is strictly less than $\omega^{\omega}$. Khoussainov and
Minnes generalized these results by proving that the height of any
word-automatic well-founded partially ordered set is below
$\omega^{\omega}$ \cite{KM10}. The proofs of these results show
that there is an
algorithm that, given a word-automatic linear order, decides  if the
linear order is an ordinal \cite{KRS05}. Moreover, there exists an
algorithm that, given two word-automatic ordinals, decides if the
ordinals are isomorphic. The decidability of the isomorphism problem
for word-automatic ordinals is obtained by extracting the Cantor
Normal Form from the given word-automatic ordinals. 
In contrast, Kuske, Liu and Lohrey proved that the isomorphism
problem for word-automatic linearly ordered sets is undecidable
\cite{Kuske2, Kuske1} and also obtained similar results for
$\omega$-automatic trees and other such structures \cite{Kuske3}.

\tinyskip

The results above naturally lead to the study of ordinals presented by
tree automata (such ordinals we call tree-automatic ordinals).
Delhomm\'e proved that any tree-automatic ordinal is strictly less
than $\omega^{\omega^{\omega}}$ \cite{D1}. Jain, Khoussainov,
Schlicht and Stephan \cite{JKSS16} connected tree-automatic
ordinals with automata working on ordinal words \cite{SS13} and
provided an alternative proof of Delhomm\'e's result. However, in
contrast to word-automatic ordinals, it is unknown if the isomorphism
problem for tree-automatic ordinals is decidable. We address this
isomorphism problem in the current paper.
We prove two main results. First, we show an ordinal $\alpha$ 
together with the ordinal addition operation $+$ is tree-automatic if
and only if $\alpha< \omega^{\omega^{\omega}}$. For the proof, by the
above mentioned result by Delhomm\'e, it suffices to show that all
ordinals $\alpha< \omega^{\omega^{\omega}}$ with the ordinal addition
operation are tree-automatic. We show that the natural
tree-representations of such ordinals preserve tree-automaticity of
the addition operation. We further provide an algorithm that, given
two tree-automatic ordinals with the addition operation, decides if
the ordinals are isomorphic. Just like in the case of word-automatic
ordinals, the proof is based on extracting Cantor Normal Forms from
tree automata that represent ordinals with addition. The Cantor Normal Form is
based on Cantor's result that every ordinal is the sum of a finite descending
list of $\omega$-powers \cite{Ca95,Ca97}.

\section{Basic definitions} 

\noindent
This section gives necessary definitions and background to tree
automata and tree-automatic structures. Let $\{0,1\}$ be the binary
alphabet and $\leq_{p}$ be the prefix order on finite binary strings.
If $x\leq_{p} y$ and $x\neq y$ then we say that $y$ properly extends
$x$. By a tree we mean a finite subset $X$ of $\{0,1\}^{\star}$ such
that (1) $X$ is closed under the prefix relation, (2) for every $x\in
X$ either no $y\in X$ properly extends $x$  or both $x0$ and $x1$
belong to $X$.  

\tinyskip

Let $X$ be  a tree and  $x\in X$. If no string in $X$ properly extends
$x$ then we say that $x$ is a leaf of the tree. In case $x\in X$ and
$x$ is not a leaf then $x$ is called an internal node of $X$; in this
case, the node $x0$ is the left-child of $x$ and $x1$ is the
right-child of $x$. For the internal node $x\in X$, both children of
$x$ belong to $X$. The null string $\lambda$ belongs to any non-empty
tree. The empty string is called the root of the tree. 
Finally, a maximal linearly ordered set of internal nodes of $X$ is
called a branch. 
\tinyskip

For a finite alphabet $\Sigma$, a $\Sigma$-tree is a function $t:
X\rightarrow \Sigma$ where $X$ is a tree. In this case the domain of
$t$,  denoted by $dom(t)$, coincides with the set $X$.  We denote the
set of all $\Sigma$-trees by $T_{\Sigma}$. A $\Sigma$-tree language
(or simply a language) is any set of $\Sigma$-trees. Thus,
$\Sigma$-tree languages are simply subsets of $T_{\Sigma}$. 

\begin{definition}
A tree automaton $\mathcal M$ is a tuple $(S, \delta, I, F)$, where
$S$ is a finite set of states, $\delta:S\times \Sigma \rightarrow
P(S\times S)$ is the transition table, $I\subseteq S$ is the set of
initial states and $F\subseteq S$ that of accepting states. These tree
automata are sometimes called top-down automata.
\end{definition}

\noindent
Given a tree automaton $\mathcal M$ and a $\Sigma$-tree $t$, one
naturally defines the notion of a  run of $\mathcal M$ on the tree
$t$. Namely, a run of $\mathcal M$ on $t$ is a  function $r:
dom(t)\rightarrow S$ such that the following conditions are satisfied: 
\begin{enumerate}
\item The run starts at the root with an initial state. Namely,
$r(\lambda)\in I$.
\item The run is consistent with the transition table. Namely, for all
internal nodes $x\in dom(t)$, if $r(x)=s$ and $t(x)=\sigma$ then
$(r(x0), r(x1))\in \delta (s,\sigma)$.  
\end{enumerate}

\noindent
Note that there could be several runs of $\mathcal M$ on any given
$\Sigma$-tree $t$.   For such a run $r$, if $r(x)\in F$ for all leaves
of the tree $dom(t)$ then we say $r$ is an accepting run of  $\mathcal
M$ on $t$. We say that $\mathcal M$ accepts the tree $t$ if $\mathcal
M$ has an accepting run on $t$.
 Define
$$
L(\mathcal M)=\{t \mid \mbox{the automaton $\mathcal M$ accepts $t$}\}.
$$
We call $L(\mathcal M)$ the language of the automaton $\mathcal M$ or
the language recognised by the automaton $\mathcal M$. 

Note that our restrictive definition of trees above simplifies the
transition table in the previous definition, since one need not
consider the case of nodes with only one direct successor. 

\begin{definition}
A $\Sigma$-tree language $L$ is called regular if there exists an
automaton $\mathcal M$ such that $L$ is the language of the automaton
$\mathcal M$, that is, $L=L(\mathcal M)$. 
\end{definition}

\noindent
It is well-known that the class of regular $\Sigma$-tree languages
forms a Boolean algebra under the set-theoretic boolean operations.
Tree automata also satisfy an important decidability property that
solves the emptiness problem.
Namely,  there exists an algorithm (that runs in linear time on the
sizes of input automata) that, given an automaton $\mathcal M$,
decides if $\mathcal M$ accepts at least one $\Sigma$-tree. 

Note that a unary relation on $T_\Sigma$ is recogised by tree automata
if and only if it is regular. 
In order to define $n$-ary relations that are recognised by tree
automata, we need to define the convolution operation on $n$-tuples of
$\Sigma$-trees. Here we describe the convolution operation for pairs
$(t_1, t_2)$ of $\Sigma$-trees; the convolution of $n$-tuples of
$\Sigma$-trees can easily be defined in a similar way. 

\begin{definition}
Let $\Diamond$ be a symbol that does not belong to $\Sigma$.  Given
$\Sigma$-trees $t_1$ and $t_2$, 
the convolution $conv(t_1,t_2)$ is a function with domain $dom(t_1)
\cup dom (t_2)$ such that for all $x\in dom(t_1) \cup dom (t_2)$ we
have the following properties:
\begin{enumerate}
\item If $x\in domt(t_1) \cap dom (t_2)$ then
$conv(t_1,t_2)(x)=(t_1(x), t_2(x))$.
\item If $x\in dom(t_1)\setminus dom(t_2)$ then 
$conv(t_1,t_2)(x)=(t_1(x), \Diamond)$.
\item If $x\in dom(t_2)\setminus dom(t_1)$ then 
$conv(t_1,t_2)(x)=(\Diamond, t_2(x))$.
\end{enumerate}
\end{definition}

\noindent
This definition allows us to define tree automata recognisable $n$-ary
relations on the set of all 
$\Sigma$-trees:
\begin{definition}
An $n$-ary relation $R$ on the set $T_{\Sigma}$ is tree automata
recognisable (or tree-automatic) if there exists
a tree automaton $\mathcal M$  that recognises the convolution 
$conv(R)$ of the relation $R$, that is:
$$
L(\mathcal M)=\{conv(t_1, t_2, \ldots, t_n) \mid (t_1, t_2, \ldots,
t_n)\in R \}.
$$
\end{definition}

\noindent
Now we connect the notion of tree automata recognisability with
algebraic structures. Recall that a relational algebraic structure is
a first order structure of the form
$$
\mathcal A=(A; R_1^{n_1}, \ldots, R_k^{n_k}),
$$ 
where $A$ is a non-empty set (called the domain of $\mathcal A$) and
for each $i=1, \ldots, k$, $R_i^{n_i}$ is a relation of arity $n_i$
on the domain $A$. These relations are called atomic relations of the
structure. Note that any structure with operation symbols can be
transformed into a relational structure by replacing each atomic
operation with its graph. So, we identify such structures with their
relational counterparts as just described.

In this paper, our interest is in ordinals with addition.
We adopt the convention to identify each ordinal with the set of
its predecessors; ordinals are further (implicitly) identified
with isomorphism types of well-founded linearly ordered sets. The ordinal
$\omega^{\omega}$ with the addition operation can be viewed as the
algebraic structure 
$$
(\omega^{\omega}; \  \leq, +),
$$
where $\omega^{\omega}$ equals the set of all ordinals strictly less than
$\omega^{\omega}$, $\leq$ is the natural linear order relation $\in$
on ordinals
(that is, $\alpha<\beta$ iff $\alpha \in \beta$ for ordinals $\alpha$ and 
$\beta$)
and $+$ is the addition operation of ordinals. Here the addition
operation on an ordinal $\delta$ is identified with its graph
$\{(\alpha, \beta, \gamma) \mid \alpha, \beta, \gamma< \delta \ \& \
\alpha+\beta=\gamma\}$. We note that not every ordinal is closed under
addition. So, $+$ is a binary partial operation.

\tinyskip

The following
definition now connects algebraic structures with tree automata: 

\begin{definition}
An algebraic structure $\mathcal A=(A; R_1^{n_1}, \ldots, R_k^{n_k})$
is tree-automatic if the domain $A\subseteq T_{\Sigma}$ and the atomic
 relations $R_1^{n_1}, \ldots, R_k^{n_k}$ on $A$ are all recognised by
tree automata.  Any tuple of tree-automata that recognise $A$,
$R_1^{n_1}$, $ \ldots$, $R_k^{n_k}$ is called a (tree-automatic)
representation of the structure.
\end{definition}

\noindent
If an algebraic structure is isomorphic to a tree-automatic structure,
then we say that it is tree-automata presentable. Since we are
interested in structures up to isomorphism, we abuse our definition
and often call tree-automata presentable structures tree-automatic
structures. This will be clear from the context. For this paper we use
the following foundational theorem in the theory of automatic
structures \cite{BG00, Ho76, Ho83, KN94}. The proof of the theorem
follows from the closure properties of regular tree-languages and the
decidability of the emptiness problem for tree automata:

\begin{theorem} \label{Thm:Decidability}
There exists an algorithm that given a tree-automatic structure
$\mathcal A$ and a first order formula $\phi(x_1, \ldots, x_n)$
with $n$ variables $x_1$, $\ldots$, $x_n$, constructs a tree automaton
$\mathcal M_{\phi}$ such that $\mathcal M_{\phi}$ accepts a tuple
$(t_1, \ldots, t_n)$ if and only if $\mathcal A\models \phi(t_1,
\ldots, t_n)$. In particular, the first order theory of any
tree-automatic structure is decidable. \qed
\end{theorem}

\noindent
Finally, we refer the reader to the articles of
Blumensath and Gr\"adel \cite{BG00}, Delhomm\'e \cite{D1},
Khoussainov and Minnes \cite{KM10}, Khoussainov and Nerode \cite{KN08},
Kuske \cite{Kuske4} and Rubin \cite{R08,R19} for the background
and open questions in the area of automatic structures. For simple and
non-trivial examples of word-automatic structures we refer to the articles of
Ishihara, Khoussainov and Rubin \cite{IKR02},
Nies \cite{Nies07}, Nies and Thomas \cite{NT08} and Stephan \cite{St15}.

\section{Tree-automatic ordinals with addition}

\noindent 
Our main result in this section is the following theorem. 

\begin{theorem}
Let $\alpha$ be an ordinal such that
$\alpha<\omega^{\omega^{\omega}}$. Then the structure $(\alpha; \ 
\leq, +)$ of the ordinal $\alpha$ together with ordinal addition is a
tree-automatic structure. 
\end{theorem}

\begin{proof}
For the reader and ease of understanding, we provide some background
and intuition for ordinals $\alpha\leq\omega^{\omega^{\omega}}$. These
will also be needed for the next theorem. The ordinal
$\omega^{\omega}$ can be viewed as the sum:
$$
\omega^{\omega}=\omega+\omega^2+\ldots + \omega^n+\ldots.
$$
The ordinals $\omega^{\omega^n}$, where $n>1$,  are defined by
induction through the supremum of the sequence of ordinals 
$$
\omega^{\omega^{n-1}},  \  \ 
\omega^{\omega^{n-1}}\cdot \  \omega^{\omega^{n-1}},  \   \ 
\omega^{\omega^{n-1}}\cdot \  \omega^{\omega^{n-1}} \cdot \ 
\omega^{\omega^{n-1}}, \  \   \   \  \omega^{\omega^{n-1}}\cdot \ 
\omega^{\omega^{n-1}} \cdot \  \omega^{\omega^{n-1}} \cdot \ 
\omega^{\omega^{n-1}}, \ \ldots.
$$ 
So, the ordinal $\omega^{\omega^{\omega}}$ can be viewed as the sum:
$$
\omega^{\omega^{\omega}}=\omega^{\omega} +
\omega^{\omega^2}+\omega^{\omega^3}+\ldots + \omega^{\omega^n}+\ldots.
$$
Let $p(X)$ be a notation for polynomials with non-negative integer
coefficients. We represent these polynomials as 
$$
p(X)= X^na_n+X^{n-1}a_{n-1}+\ldots+ Xa_1+a_0,
$$ 
where \  $a_n>0$ whenever $p(X) \neq 0$.
Here we explicitly wrote the coefficients on the right side of the
variables, since the multiplication and addition operations on
ordinals are not commutative. For any ordinal
$\alpha<\omega^{\omega^{\omega}}$, there are unique polynomials
$p_0(X)$, $\ldots$, $p_k(X)$ 
and integer coefficients $c_0$, $\ldots$, $c_k$ with $c_k>0$  
such that 
\begin{itemize}
\item
$\alpha=\omega^{p_0(\omega)}c_{0}+\omega^{p_1(\omega)}c_{1}+\ldots +
\omega^{p_{k-1}(\omega)}c_{k-1}+ \omega^{p_k(\omega)}c_k$ 
and 
\item $p_0(\omega)> p_{1}(\omega)> \ldots> p_k(\omega)$. 
\end{itemize}
When adding ordinals described in the form above, one takes into
account the following equations from ordinal arithmetic: 
$$
\omega^{\alpha}m + \omega^{\alpha}n=\omega^{\alpha}(m+n), \  \mbox{and
} \ \omega^{\alpha}+\omega^{\beta}=\omega^{\beta},
$$
where $m$ and $n$ are natural numbers and $\alpha< \beta$ are
ordinals. For instance,
$$
(\omega^{\omega^{3}}4+ \omega^{\omega^{2}}7+\omega^{6}3+ \omega^{2}
+1) + (\omega^{\omega^{2}}2+\omega^{6}3+ \omega^{5}+5)=
\omega^{\omega^{3}}4+ \omega^{\omega^{2}}9+\omega^{6}3+ \omega^{5}+5.
$$

\noindent
To prove the desired result, it suffices 
to show that each ordinal $\omega^{\omega^n}$ with the ordinal
addition operation $+$ is a tree-automatic structure. 
This yields a representation of $(\alpha;\ \leq,+)$ for any ordinal
$\alpha<\omega^{\omega^\omega}$, since the class of tree-automatic
structures is closed under products. The proof is by induction on $n$. 

We first explain a tree-automatic presentation of $\omega^{\omega}$.
Note that every non-null ordinal $\alpha< \omega^{\omega}$ can be 
uniquely written as 
$$ 
\omega^{m}b_m+\ldots + \omega^{1} b_1+\omega^0 b_0, 
$$ 
where $b_m>0$ and the coefficients $b_m, \ldots, b_0$ are natural
numbers. The integer $m$ is called the degree of $\alpha$.   We
represent this $\alpha$ as the following $\{0,1\}$-tree $t_{\alpha}$: 
\begin{enumerate} 
\item \label{one} 
The leftmost branch of $t_{\alpha}$ has length $m+1$ and nodes
$v_i=0^i$ for $i\leq m$. 
\item \label{two} 
For each $i$, the rightmost branch containing $v_i$ has length $k_i$
above $v_i$ and nodes $w_j=v_i^\smallfrown 1^j$ for $j\leq k_i$ and
codes the coefficient $b_i\in [2^{k_i},2^{k_i+1})$ in binary format,
beginning with the least-significant bit.
\item 
$dom(t_\alpha)$ does not contain any more nodes than implied by
(\ref{one}) and (\ref{two}), and $t_\alpha(s)=0$ for all $s\in
dom(t_\alpha)$ not mentioned there. 
\end{enumerate}
Let $L_1$ be the set of all $\{0,1\}$-trees  $t_{\alpha}$ that
represent ordinals $\alpha< \omega^{\omega}$. 
It is clear that $L_1$ is a regular tree language. 
It is not too hard to  see that there exist tree automata $\mathcal
M_{\leq}^{(1)}$ and $\mathcal M_{+}^{(1)}$ such that, given trees
$t_{\alpha}$, $t_{\beta}$ and $t_{\gamma}$ representing ordinals
$\alpha, \beta, \gamma < \omega^{\omega}$, the automata $\mathcal
M_{\leq}^{(1)}$ and $\mathcal M_{+}^{(1)}$ verify that $\alpha\leq
\beta$ and $\alpha+\beta=\gamma$.  For instance, $\mathcal
M_{\leq}^{(1)}$ by reading $conv(t_{\alpha}, t_{\beta})$
non-deterministically makes one of the following choices: 
\begin{enumerate} 
\item The automaton guesses and verifies that the degree 
of $\beta$ is greater than that of $\alpha$. 
In this case $\mathcal M_{\leq}^{(1)}$ accept the convolution
$conv(t_{\alpha}, t_{\beta})$. 
\item The automaton guesses and verifies that 
both $\alpha$ and $\beta$ have the same degree, 
and then compares the coefficients
of $\alpha$ and $\beta$ lexicographically. 
\end{enumerate}
Similarly, using the properties of the addition operation stated
above, one can describe the desired automaton  $\mathcal M_{+}^{(1)}$. 

\begin{enumerate} 
\item \label{one'} 
The automaton guesses at each internal node $v_i$ of $dom(t_\alpha)$
whether the Cantor Normal Form of $\beta$ contains a non-vanishing
multiple of $\omega^k$ for some $k>i$. 
\item 
For any $i<k$ such that the guess in (\ref{one'}) is positive, we take
as coefficient of $\omega^i$ in the sum only the one appearing in
the representation of $\beta$;  
when the guess is negative, we take the sum of the coefficients of $\omega^i$
appearing in $\alpha$ and $\beta$.
\end{enumerate} 

\noindent
So, $\omega^{\omega}$ with the addition operation is a tree-automatic
structure. 

Suppose now that we have a regular tree language $L_n$, that represents
all ordinals $\alpha< \omega^{\omega^n}$, and tree automata $\mathcal
M_{\leq}^{(n)}$ and $\mathcal M_{+}^{(n)}$ such that, given trees
$t_{\alpha}$, $t_{\beta}$ and $t_{\gamma}$ representing ordinals
$\alpha, \beta, \gamma < \omega^{\omega^n}$, the automata $\mathcal
M_{\leq}^{(n)}$ and $\mathcal M_{+}^{(n)}$ verify that $\alpha\leq
\beta$ and $\alpha+\beta=\gamma$.  

Elements $\alpha$ of  $\omega^{\omega^{n+1}}$ can be identified with 
finite tuples $(\alpha_0, \ldots, \alpha_k)$, where each 
$\alpha_i$, $i=1, \ldots,k$,  is an element of  $\omega^{\omega^n}$. 
Intuitively, the tuple $(\alpha_0,\ldots,\alpha_k)$ represents
the ordinal 
$(\omega^{\omega^n})^k \times (\alpha_k)+
(\omega^{\omega^n})^{k-1} \times \alpha_{k-1}+\ldots+
(\omega^{\omega^n}) \times \alpha_{1}+
(\omega^{\omega^n})^0 \times \alpha_{0}$.
The order between such tuples is given by length-lexicographic order
of $(\alpha_k,\ldots,\alpha_0)$, when $\alpha_k$ is non-zero or $k=0$. 
We represent these tuples $\alpha$ as the following binary tree:
\begin{enumerate}
\item 
The leftmost branch of $t_{\alpha}$ has length $k+1$ and nodes
$v_i=0^i$ for $i\leq k$. 
\item For each $i\leq k$,  the right subtree of $v_i$ is a copy of 
$t_{\alpha_i}$ with its root copied to the right child of $v_i$. 
\item 
$dom(t_\alpha)$ does not contain any more nodes than implied by
(\ref{one}) and (\ref{two}) and $t_\alpha(s)=0$ for all $s\in
dom(t_\alpha)$ not mentioned there. 
\end{enumerate} 

\noindent
Let $L_{n+1}$ be the tree language consisting of trees as we described
above. The tree language $L_{n+1}$ is a regular tree language. So, the
trees from $L_{n+1}$ represent elements of the ordinal
$\omega^{\omega^{n+1}}$. Now, using the tree automata $\mathcal
M_{\leq}^{(n)}$ and $\mathcal M_{+}^{(n)}$  as subroutines, it is not
too hard to construct tree automata $\mathcal M_{\leq}^{(n+1)}$ and
$\mathcal M_{+}^{(n+1)}$ that recognise the order relation and the
addition  operation of the ordinal $\omega^{\omega^{n+1}}$. The proof
is similar to the case of  $\mathcal M_{\leq}^{(1)}$ and  $\mathcal
M_{+}^{(1)}$ above.
\end{proof}

\begin{remark}
The anonymous referees asked whether there exist ordinals $\alpha$
for which $(\alpha;\leq,+)$ is not automatic while $(\alpha;\leq)$ is
automatic. Of course, if there are $\beta,\gamma < \alpha$ with
$\beta+\gamma \geq \alpha$ then the ordinal $\alpha$ can treat
the addition only as an automatic relation $\{(\beta,\gamma,\delta):
\beta,\gamma,\delta < \alpha \wedge \beta+\gamma = \delta\}$
and not as a function, as the latter would have undefined
places in the case that the sum exceeds $\alpha$. For these
ordinals, one has that if $(\alpha;\leq)$ has an automatic presentation
then $(\alpha;\leq,+)$ has also an automatic presentation, which might
be different from the previous one, as
seen above. However, if $\omega \leq \alpha < \omega^{\omega^\omega}$,
there is a further presentation where $(\alpha;\leq)$ is automatic but
the addition not. The main idea would be to code the coefficients of
the Cantor normal form not in binary as done above, but in unary;
then there are trees where the height of the sum is approximately
twice that of the original ordinals and an easy application of the
pumping lemma shows that such a function cannot be tree-automatic.
The natural numbers coded in unary are a well-known example of a
word-automatic structure where order and remainders are automatic
while the addition is not.
\end{remark}

\section{Deciding the isomorphism problem}

\noindent
The goal of this section is to prove the following decision theorem
that solves the isomorphism problem for tree-automatic ordinals with
the addition operation. 
For this, we need to handle addition on sets which are not closed
under addition.
Therefore we recall that $(\alpha;+,<)$ is a tree-automatic ordinal
with addition
iff $(\alpha;<)$ is a tree-automatic ordinal and $+$ is a partial
tree-automatic function with tree-automatic domain such that for $\beta,\gamma
\in \alpha$, $\beta+\gamma$ is defined and takes as value the ordinal sum of
$\beta$ and $\gamma$ iff that sum is a member of $\alpha$, that is, iff it is
strictly less than $\alpha$.

\begin{theorem}
There exists an algorithm that, given two tree automatic ordinals with
the addition operation, decides if the ordinals are isomorphic. 
\end{theorem}

\begin{proof}
Recall that, by Delhomm\'e's result mentioned in the introduction,  
if $(\alpha; \leq, +)$ is a tree-automatic ordinal, then $\alpha <
\omega^{\omega^{\omega}}$. We will be using this fact implicitly. The
following well-known claim (see e.g. \cite{J03,S58}) shows that we can use
the addition operation $+$ to identify powers of the ordinal $\omega$. 
For the present work, we only consider such powers of the ordinal
$\omega$ which are of the form $\omega^{p(\omega)}$, where
$p(X)$ is a polynomial; larger powers of $\omega$ like
$\omega^{\omega^\omega}$ do not need to be considered, as they
are not tree-automatic.

\begin{claim} \label{powers of omega} 
An ordinal $\beta<\omega^{\omega^{\omega}}$ is closed under
the addition operation $+$ if and only if
$\beta$ is a power of the ordinal $\omega$. 
\end{claim}

\noindent
Indeed, if  $\beta$ is of the form $\omega^{p(\omega)}$, where $p(X)$
is a polynomial, then it is easy to check that $\beta$ is closed under
the addition operation $+$. Otherwise, $\beta$ is of the form
$$
\beta=\omega^{p(\omega)}c+\beta',
$$
where $0<\beta'< \omega^{p(\omega)}$ and $c>0$, or $\beta'=0$ and $c>1$. So,
$\omega^{p(\omega)}<\beta$ yet $\omega^{p(\omega)} (c+1)\geq\beta$.
Therefore, $\beta$ is not closed under the addition operation. 

\tinyskip

So, let $(\alpha; \  \leq, +)$ be a tree-automatic ordinal
with the addition operation $+$. Consider the following set:
$$
P_{\alpha}=\{\beta \leq \alpha \mid \beta\neq0\ \&\ \forall \gamma
\forall \gamma' (\gamma<\beta \ \& \ \gamma'< \beta \rightarrow
\gamma+\gamma' \mbox{ is defined 
and } \gamma+\gamma'<\beta)\}. 
$$
The set $P_{\alpha}$ is a regular tree language by Theorem
\ref{Thm:Decidability} (since it is first order definable) and the
order $\leq$ restricted to $P_{\alpha}$ is an ordinal; we will
identify $P_\alpha$ with this ordinal. 
Our second claim is the following:

\begin{claim} \label{bound for Palpha} 
The ordinal $P_{\alpha}$ is a tree-automatic ordinal strictly less
than $\omega^{\omega}$.
\end{claim}

\noindent
To see this, use that the set $P_\alpha$ consists of the powers of
$\omega$ less or equal to $\alpha$ by Claim \ref{powers of omega}.
If the ordinal $P_\alpha$ equals $\gamma$, then $\omega^\gamma\leq\alpha$. 
Since $\alpha<\omega^{\omega^\omega}$, this implies $\gamma<\omega^\omega$. 

\begin{claim} \label{compute Cantor normal forma} 
Given a tree-automatic ordinal $\alpha$, we can effectively compute
the Cantor Normal Form of the ordinal $P_{\alpha}$.
\end{claim}

\noindent
Here our argument is the same as in the word-automatic case for
ordinals \cite{KRS05}. 

By Theorem \ref{Thm:Decidability}, there is an algorithm that, when
given as input any tree-automatic presentation of $(\alpha;+,\leq)$,
computes a tree-automatic presentation of the ordinal $(P_\alpha;\leq)$. 
Let $\gamma_0$ denote this ordinal. It is strictly less than
$\omega^\omega$ by Claim \ref{bound for Palpha}. 

We will iterate the following step, beginning with $k=0$, until the
resulting ordinal is the empty set.
First we compute the number $b_k$ of steps which one needs to remove
the largest element of $\gamma_k$ until the resulting ordinal is a limit
ordinal, that is, has no largest element. Next we let $\gamma_{k+1}$ be the
subset of $\gamma_k$ of all limit ordinals in $\gamma_k$, that is, of all
ordinals $\delta$ which satisfy that there are smaller ordinals than $\delta$,
but among those there is no largest one. Now we let $k=k+1$ and iterate this
process.

For each $k$ it holds that $\gamma_k = \omega \cdot \gamma_{k+1} +b_k$.
Thus, in the Cantor Normal Form, the degree of $\gamma_{k+1}$ is one below
that of $\gamma_k$ and $b_k$ is the last constant part of the Cantor Normal
Form of $\gamma_k$. Iterating this gives that the Cantor Normal Form of
$P_\alpha$ is $\omega^h\cdot b_h+\ldots+\omega\cdot b_1+b_0$ 
and $h$ is the least value where $\gamma_{h+1}$ is the empty set.

Since the procedure is effective by Theorem \ref{Thm:Decidability},
this proves Claim \ref{compute Cantor normal forma}. 

\begin{claim} \label{compute Cantor normal formb} 
Given a tree-automatic ordinal with addition $(\alpha;\leq,+)$, 
we can effectively compute a tree-automatic ordinal with addition
$(\alpha';\leq,+)$
such that $\alpha = \omega^{P_\alpha}+\alpha'$ and this procedure allows
to compute the Cantor Normal Form of $\alpha$.
\end{claim}

\noindent
The ordinal $\omega^{P_{\alpha}}$ is equal to $\alpha$ in the case
that $\alpha$ is closed under $+$ and in this case, $\alpha' = 0$.
Otherwise $\omega^{P_{\alpha}}$
is the smallest ordinal $\beta \in \alpha$ such that there is no
$\gamma$ with $\beta < \gamma < \alpha$ which is closed under $+$;
this $\beta$ exists by the well-orderedness of $\alpha$.
Now we define $\alpha' = \{\gamma \in \alpha: \beta + \gamma$ is
defined and in $\alpha\}$
and inherit the addition from $\alpha$ to $\alpha'$ as follows:
for $\gamma,\delta \in \alpha'$, we check whether $\gamma+\delta$ and
$\beta+(\gamma+\delta)$ are both defined in $(\alpha;\leq,+)$; if so,
then we let $\gamma+\delta$ have in $(\alpha';\leq,+)$ the same value
as in $(\alpha;\leq,+)$, else we let $\gamma+\delta$ be undefined
in $(\alpha';\leq,+)$. We note that
$\alpha' < \alpha$, as $P_\alpha$ is the exponent of the highest
$\omega$-power with nonzero coefficient in the Cantor Normal Form
of $\alpha$ and this coefficient is then in $\alpha'$ at least by one smaller.

Now we let $\alpha_0 = \alpha$ and $\ell=0$.
While $\alpha_\ell$ is not an $\omega$-power,
we update $\alpha_{\ell+1} = \alpha_\ell'$ and $\ell = \ell+1$.
Now let $\ell$ be the final value of this process and
$\alpha_0,\alpha_1,\ldots,\alpha_\ell$ as defined on the way.
It is easy to see by induction that the ordinals
in this sequence form a strictly descending chain and that
$$
   \alpha = \sum_{k=0}^\ell \omega^{P_{\alpha_k}} 
  = \omega^{P_{\alpha_0}}+
   \omega^{P_{\alpha_1}}+\ldots+\omega^{P_{\alpha_\ell}}.
$$
This is just due to the fact that for the Cantor Normal Form 
$\sum_{k=0}^{h} \omega^{\beta_k}$ of $\alpha$, by induction
for every $k$ the Cantor Normal Form of $\alpha_k$ equals 
$\sum_{k'=k}^h \omega^{\beta_{k'}}$ and thus
$P_{\alpha_k} = \beta_k$. It follows that $\ell = h$,
since $h$ is least such that $\omega^{\beta_h}$ is an $\omega$-power. 
We can effectively determine $\ell$ and the Cantor Normal Form of
all $P_{\alpha_k}$ by Claim \ref{compute Cantor normal forma}. 
We can therefore also effectively determine the 
Cantor Normal Form of $\alpha$ by counting equal $\omega$-powers in
the last equation. 

\begin{claim}
The isomorphism problem of tree-automatic ordinals with addition is decidable.
\end{claim}

\noindent
This follows from Claim \ref{compute Cantor normal formb}. 
Suppose that two tree-automatic ordinals with addition
$(\alpha;\leq,+)$ and $(\beta;\leq,+)$ are given. Using the process above, 
we produce the Cantor Normal Forms for both of these ordinals. 
The ordinals are isomorphic if and only if the two Cantor Normal Forms
produced coincide. 
\end{proof}

\noindent
{\bf Acknowledgments.} The authors would like to thank Dietrich Kuske
for very helpful correspondence, Peter Koepke for discussions and
the anonymous referees for their helpful comments. Part of
this work was done while F.~Stephan visited the University of Bonn during
his sabbatical leave and while P.~Schlicht was working there.
S.~Jain was supported in part by NUS grant C252-000-087-001;
B.~Khoussainov was supported in part
by the Marsden Fund grant of the Royal Society of New Zealand;
S.~Jain, B.~Khoussainov and F.~Stephan have been supported in
part by the Singapore Ministry of Education Academic Research Fund Tier 2
grant MOE2016-T2-1-019 / R146-000-234-112; furthermore,
P.~Schlicht was supported by a Marie Sk\l odowska-Curie Individual
Fellowship with number 794020.

\end{document}